\documentclass[
preprint,
amsmath,amssymb,
aps,
]{revtex4-2}

\usepackage{graphicx}
\usepackage{dcolumn}
\usepackage{bm}
\usepackage{hyperref}

\usepackage{amsthm}

\usepackage{dcolumn}
\usepackage{bm}
\usepackage{hyperref}
\usepackage{subcaption}
\usepackage{braket}
\usepackage{caption}
\usepackage{diagbox}
\usepackage{booktabs}
\usepackage{adjustbox}
\usepackage{multirow}
\usepackage[mathlines]{lineno}
\usepackage{ragged2e}
\usepackage{pifont}
\usepackage{makecell}
\usepackage{algpseudocode}
\DeclareCaptionType{algorithm}[Algorithm][List of Algorithms]
\usepackage{soul}
\usepackage{cleveref}
\usepackage{tikz}
\usepackage{array}
\usepackage{centernot}
\usepackage{cancel}
\usepackage{enumitem}

\newtheorem{theorem}{Theorem}
\newtheorem{lemma}{Lemma}
\newtheorem{proposition}{Proposition}
\newtheorem{corollary}{Corollary}

\newcommand{\R}{\mathbb{R}}

\newcommand{\Tr}{\mathrm{tr}}

\newcommand{\calM}{\mathcal{M}}
\newcommand{\justO}{\operatorname{O}}
\newcommand{\SO}{\operatorname{SO}}

\newcommand{\diag}{\operatorname{diag}}

\definecolor{codegreen}{rgb}{0,0.6,0}
\definecolor{codegray}{rgb}{0.5,0.5,0.5}
\definecolor{codepurple}{rgb}{0.58,0,0.82}
\definecolor{backcolour}{rgb}{0.95,0.95,0.92}

\begin{document}


\title{CHSH inequality always holds in bipartite qutrits\\with spin-1 observables}

\author{Hyunho Cha}
\email{Contact author: ovalavo@snu.ac.kr}
\author{Jungwoo Lee}
\email{Contact author: junglee@snu.ac.kr}
\affiliation{NextQuantum Innovation Research Center, Seoul National University, Seoul 08826, Republic of Korea}


\begin{abstract}
Bell inequalities reveal when quantum correlations cannot be explained by any local hidden-variable model, and the CHSH inequality is the most fundamental test of this phenomenon. While entangled two-qubit states can violate CHSH inequalities with suitable local measurements, the situation is more subtle for higher-dimensional systems when the allowed observables are physically restricted. In this work, we study bipartite qutrit systems in which each party is restricted to spin-1 observables \(S(u)=u_xS_x+u_yS_y+u_zS_z\), with \(u\in\mathbb R^3\) a unit vector. Hanotel and Loubenets conjectured that nonseparable pure states of two qutrits do not violate the CHSH inequality under this measurement restriction. We prove the stronger result that every bipartite state on \(\mathbb C^3\otimes\mathbb C^3\), pure or mixed, satisfies the CHSH inequality for all choices of local spin-1 observables. Thus the absence of CHSH violation is not a property of special states, but a structural limitation of the spin-1 measurement family. Our result clarifies a distinction between two-qutrit entanglement and CHSH nonlocality under natural spin-component measurements.
\end{abstract}

\maketitle


\section{Introduction}

Bell inequalities provide a fundamental way to distinguish correlations that can be explained by local hidden-variable models from those allowed by quantum theory \citep{bell1964einstein, freedman1972experimental, clauser1974experimental, clauser1978bell, aspect1981experimental, aspect1982experimental, fine1982hidden, werner2001bell, genovese2005research, brunner2014bell, hensen2015loophole}. Among them, the CHSH inequality is the most basic and widely studied example \citep{clauser1969proposed, landau1988empirical, popescu1994quantum, wehner2006tsirelson, mckague2012robust, kaniewski2016analytic}. In the standard two-qubit setting, suitable choices of local observables can lead to a value strictly larger than the classical bound \citep{cirel1980quantum, gisin1991bell, horodecki1995violating}. This naturally raises the question of how the possibility of CHSH violation depends on both the local Hilbert spaces and the class of observables that the parties are allowed to measure \citep{mermin1980quantum, gisin1992maximal, peres1992finite, loubenets2020generalized, peruzzo2023entanglement, loubenets2024quantifying, loubenets2025high}.

This work concerns a restricted but physically natural two-qutrit setting \citep{ho1998spinor, ohmi1998bose, howell2002experimental}. Each local system is described by \(\mathbb C^3\), and each party measures spin-\(1\) observables of the form \(S(u)=u_xS_x+u_yS_y+u_zS_z\), where \(u=(u_x,u_y,u_z)\in\R^3\) is a unit vector. Such observables have spectrum \(\{-1,0,1\}\), and hence are bounded in the range relevant to the CHSH expression, but they form a special three-parameter family among all Hermitian observables on \(\mathbb C^3\). Therefore, even for entangled states on \(\mathbb C^3\otimes\mathbb C^3\), it is not immediate whether this restricted family of spin-\(1\) measurements is sufficient to produce a CHSH violation.

Hanotel and Loubenets conjectured that nonseparable pure states of two qutrits do not violate the CHSH inequality when the local observables are restricted to spin-\(1\) observables \citep{hanotel2025nonviolation}. We prove a stronger statement. Namely, for every density operator \(\rho\) on \(\mathbb C^3\otimes\mathbb C^3\) and every choice of unit vectors \(a,a',b,b'\in\R^3\), the corresponding CHSH Bell operator \(\mathcal B(a,a',b,b')\) satisfies \(|\Tr(\rho\,\mathcal B(a,a',b,b'))|\le 2\). Thus the conclusion holds not only for nonseparable pure states, but for all bipartite qutrit states, including mixed states.

Thus, CHSH nonviolation in this setting is a structural limitation imposed by the spin-1 measurement family. Although two-qutrit entanglement can exhibit nonclassical features under more general measurements \citep{acin2002quantum, chen2002wringing, kaszlikowski2002clauser, fu2003maximal, gruca2012nonclassicality}, the observables \(S(u)\) are too restricted to reveal CHSH nonlocality. In this sense, the result separates entanglement from CHSH violation in a particularly transparent physical model \citep{werner1989quantum, verstraete2001entanglement}.

\section{Setup and main result}
Let
\begin{align*}
S_x = \frac{1}{\sqrt2}
\begin{pmatrix}
0&1&0\\
1&0&1\\
0&1&0
\end{pmatrix},\qquad
S_y=\frac{1}{\sqrt2}
\begin{pmatrix}
0&-i&0\\
i&0&-i\\
0&i&0
\end{pmatrix},\qquad
S_z =
\begin{pmatrix}
1&0&0\\
0&0&0\\
0&0&-1
\end{pmatrix}
\end{align*}
be the spin-$1$ matrices on $\mathbb C^3$. For a unit vector $u=(u_x,u_y,u_z)\in \R^3$, define
\[
S(u)=u_xS_x+u_yS_y+u_zS_z.
\]
For each $R\in \SO(3)$, let $U_R$ denote the spin-$1$ representation of $R$ \citep{kawaguchi2012spinor}, so that
\[
U_R S(u) U_R^{\dagger}=S(Ru)
\qquad (u\in \R^3).
\]
In particular, if $\|u\|=1$, then $S(u)$ is unitarily equivalent to $S_z$, hence has spectrum $\{-1,0,1\}$.

Let $\rho$ be a density operator on $\mathbb C^3\otimes \mathbb C^3$, and let $a,a',b,b'\in \R^3$ be unit vectors. The associated CHSH Bell operator is
\begin{align*}
\mathcal B(a,a',b,b')
= S(a)\otimes S(b)
+ S(a)\otimes S(b')
+ S(a')\otimes S(b)
- S(a')\otimes S(b').
\end{align*}

Hanotel and Loubenets \citep{hanotel2025nonviolation} conjectured that nonseparable pure states of two qutrits do not violate the CHSH inequality when each party is restricted to spin-1 observables. The following result establishes a stronger statement that no bipartite qutrit state violates the CHSH inequality under such measurements.

\begin{theorem}\label{thm:main}
For every density operator $\rho$ on $\mathbb C^3\otimes \mathbb C^3$ and every choice of unit vectors $a,a',b,b'\in \R^3$,
\[
\bigl|\Tr\bigl(\rho\,\mathcal B(a,a',b,b')\bigr)\bigr|\le 2.
\]
Equivalently, no two-qutrit state violates the CHSH inequality under spin-$1$ measurements.
\end{theorem}

The proof will be obtained by reducing $\mathcal B(a,a',b,b')$ to a two-parameter operator whose spectrum can be computed exactly.

\section{Reduction to a two-parameter model}
For $M \in \calM_3(\R)$, define
\[
K(M):=\sum_{i,j=1}^3 M_{ij}\,S_i\otimes S_j,
\]
where $(S_1,S_2,S_3)=(S_x,S_y,S_z)$. Then
\[
\mathcal B(a,a',b,b') = K(M),
\qquad
M:=a(b+b')^{\!T}+a'(b-b')^{\!T}.
\]
Since $M$ is a sum of two rank-one matrices, $\operatorname{rank}(M)\le 2$.

\begin{lemma}[Rotational covariance]\label{lem:covariance}
For any $R,Q\in \SO(3)$ and any $M\in \calM_3(\R)$,
\[
(U_R\otimes U_Q)\,K(M)\,(U_R\otimes U_Q)^{\dagger}=K(RMQ^T).
\]
In particular, $K(M)$ and $K(RMQ^T)$ are unitarily equivalent and therefore have the same spectrum and operator norm.
\end{lemma}

\begin{proof}
Using $U_RS_iU_R^{\dagger}=\sum_{k=1}^3R_{ki}S_k$ and $U_QS_jU_Q^{\dagger}=\sum_{\ell=1}^3Q_{\ell j}S_{\ell}$, we obtain
\begin{align*}
(U_R\otimes U_Q)K(M)(U_R\otimes U_Q)^{\dagger}
&=\sum_{i,j=1}^3 M_{ij}\,(U_RS_iU_R^{\dagger})\otimes (U_QS_jU_Q^{\dagger})\\
&=\sum_{i,j=1}^3 M_{ij}\Bigl(\sum_{k=1}^3R_{ki}S_k\Bigr)\otimes\Bigl(\sum_{\ell=1}^3Q_{\ell j}S_{\ell}\Bigr)\\
&=\sum_{k,\ell=1}^3 (RMQ^T)_{k\ell}\,S_k\otimes S_{\ell}\\
&=K(RMQ^T).
\end{align*}
\end{proof}

\begin{lemma}\label{lem:svd}
Let $M\in \calM_3(\R)$ have rank at most $2$. Then there exist $R,Q\in \SO(3)$ and singular values $s,t\ge 0$ of $M$ such that
\[
RMQ^T=\diag(s,0,t).
\]
Consequently, $K(M)$ is unitarily equivalent to
\[
H_{s,t}:=s\,S_x\otimes S_x+t\,S_z\otimes S_z.
\]
\end{lemma}

\begin{proof}
Ordinary singular-value decomposition gives matrices $O_1,O_2\in \justO(3)$ and numbers $s,t\ge 0$ such that
\[
O_1MO_2^T = D:=\diag(s,t,0).
\]
Let
\[
J:=\diag(1,1,-1).
\]
Because the third singular value is zero, one has $JD=D$ and $DJ=D$.
If $\det(O_1)=-1$, replace $O_1$ by $JO_1$. If $\det(O_2)=-1$, replace $O_2$ by $JO_2$.
These replacements do not change the diagonal form because left or right multiplication by $J$ leaves $D$ unchanged.
Hence we may choose $R_0,Q_0\in \SO(3)$ such that
\[
R_0MQ_0^T=\diag(s,t,0).
\]
Now let
\[
P=
\begin{pmatrix}
1&0&0\\
0&0&1\\
0&-1&0
\end{pmatrix}\in \SO(3).
\]
A direct computation gives
\[
P\,\diag(s,t,0)\,P^T=\diag(s,0,t).
\]
Setting $R:=PR_0$ and $Q:=PQ_0$, we obtain $R,Q\in \SO(3)$ and
\[
RMQ^T = P(R_0MQ_0^T)P^T = \diag(s,0,t).
\]
By Lemma~\ref{lem:covariance}, $K(M)$ is unitarily equivalent to
\[
K(\diag(s,0,t))=s\,S_x\otimes S_x+t\,S_z\otimes S_z,
\]
as claimed.
\end{proof}

\begin{lemma}\label{lem:sumsq}
For the matrix
\[
M=a(b+b')^{\!T}+a'(b-b')^{\!T},
\]
if $s,t\ge 0$ are its first two singular values, then
\[
s^2+t^2=4.
\]
\end{lemma}

\begin{proof}
Set
\[
c:=b+b',\qquad d:=b-b'.
\]
Then
\[
c\cdot d = \|b\|^2-\|b'\|^2 = 1-1 = 0.
\]
Moreover,
\[
M=ac^T+a'd^T.
\]
Using the Frobenius inner product $\langle X,Y\rangle_\mathrm{F}:=\Tr(X^TY)$, together with
\[
\|uv^T\|_\mathrm{F}^2 = \|u\|^2\|v\|^2,
\qquad
\langle uv^T,xy^T\rangle_\mathrm{F}=(u\cdot x)(v\cdot y),
\]
we get
\begin{align*}
\|M\|_\mathrm{F}^2
&=\|ac^T\|_\mathrm{F}^2+\|a'd^T\|_\mathrm{F}^2+2\langle ac^T,a'd^T\rangle_\mathrm{F}\\
&=\|a\|^2\|c\|^2+\|a'\|^2\|d\|^2+2(a\cdot a')(c\cdot d)\\
&=\|c\|^2+\|d\|^2.
\end{align*}
Since $b,b'$ are unit vectors,
\begin{align*}
\|c\|^2+\|d\|^2
&=\|b+b'\|^2+\|b-b'\|^2\\
&=(2+2b\cdot b')+(2-2b\cdot b')\\
&=4.
\end{align*}
On the other hand, the singular values of $M$ are $s,t,0$, so
\[
\|M\|_\mathrm{F}^2=s^2+t^2.
\]
Therefore $s^2+t^2=4$.
\end{proof}

\section{Exact spectrum of the reduced operator}
We now compute the spectrum of
\[
H_{s,t}=s\,S_x\otimes S_x+t\,S_z\otimes S_z.
\]
Let $\{\ket{1},\ket{0},\ket{-1}\}$ be the $S_z$-eigenbasis, so that
\[
S_z\ket{m}=m\ket{m},
\qquad m\in\{1,0,-1\},
\]
and
\begin{align*}
S_x\ket{1} & =\frac{1}{\sqrt2}\ket{0},
\qquad
S_x\ket{0}=\frac{1}{\sqrt2}(\ket{1}+\ket{-1}),\\
S_x\ket{-1}&=\frac{1}{\sqrt2}\ket{0}.
\end{align*}
Write $\ket{m,n}:=\ket{m}\otimes\ket{n}$.

Consider the orthogonal decomposition
\[
\mathbb C^3\otimes\mathbb C^3 = V_4\oplus V_5,
\]
where
\[
V_4:=\operatorname{span}\{\ket{1,0},\ket{0,1},\ket{0,-1},\ket{-1,0}\}
\]
and
\[
V_5:=\operatorname{span}\{\ket{1,1},\ket{1,-1},\ket{0,0},\ket{-1,1},\ket{-1,-1}\}.
\]
A direct computation shows that both subspaces are invariant under $H_{s,t}$.

\begin{proposition}\label{prop:spectrum}
The spectrum of $H_{s,t}$ is
\[
\operatorname{spec}(H_{s,t})=
\bigl\{0,0,0,\pm s,\pm t,\pm\sqrt{s^2+t^2}\bigr\}.
\]
Hence
\[
\|H_{s,t}\|=\sqrt{s^2+t^2}.
\]
\end{proposition}

\begin{proof}
First restrict to $V_4$. In the ordered basis
\[
\ket{1,0},\ \ket{0,1},\ \ket{0,-1},\ \ket{-1,0},
\]
one finds
\[
H_{s,t}\big|_{V_4}=
\begin{pmatrix}
0&s/2&s/2&0\\
s/2&0&0&s/2\\
s/2&0&0&s/2\\
0&s/2&s/2&0
\end{pmatrix}.
\]
Its characteristic polynomial is
\[
\lambda^2(\lambda^2-s^2),
\]
so the eigenvalues on $V_4$ are
\[
0,\ 0,\ s,\ -s.
\]

Now restrict to $V_5$. Define
\[
u_1:=\ket{1,1}-\ket{-1,-1},
\qquad
u_2:=\ket{1,-1}-\ket{-1,1}.
\]
A direct computation gives
\[
H_{s,t}u_1=tu_1,
\qquad
H_{s,t}u_2=-tu_2.
\]
On the remaining three-dimensional subspace, use the orthonormal basis
\begin{align*}
w_1 :=\frac{\ket{1,1}+\ket{-1,-1}}{\sqrt2},
\qquad
w_2:=\frac{\ket{1,-1}+\ket{-1,1}}{\sqrt2},\qquad
w_3 :=\ket{0,0}.
\end{align*}
In this basis,
\[
H_{s,t}\big|_{\operatorname{span}\{w_1,w_2,w_3\}}=
\begin{pmatrix}
t&0&s/\sqrt2\\
0&-t&s/\sqrt2\\
s/\sqrt2&s/\sqrt2&0
\end{pmatrix}.
\]
Its characteristic polynomial is
\[
\lambda\bigl(\lambda^2-s^2-t^2\bigr),
\]
so the eigenvalues on this block are
\[
0,\ \sqrt{s^2+t^2},\ -\sqrt{s^2+t^2}.
\]
Combining the two invariant subspaces yields
\[
\operatorname{spec}(H_{s,t})=
\bigl\{0,0,0,\pm s,\pm t,\pm\sqrt{s^2+t^2}\bigr\}.
\]
Since $\sqrt{s^2+t^2}\ge s,t$, we have
\[
\|H_{s,t}\|=\sqrt{s^2+t^2}.
\]
\end{proof}

\section{Proof of Theorem~\ref{thm:main}}
\begin{proof}[Proof of Theorem~\ref{thm:main}]
Let
\[
M:=a(b+b')^{\!T}+a'(b-b')^{\!T}.
\]
Then $\mathcal B(a,a',b,b')=K(M)$. By Lemma~\ref{lem:svd}, there exist $s,t\ge 0$ such that $\mathcal B(a,a',b,b')$ is unitarily equivalent to
\[
H_{s,t}=s\,S_x\otimes S_x+t\,S_z\otimes S_z.
\]
By Lemma~\ref{lem:sumsq}, one has $s^2+t^2=4$. Hence Proposition~\ref{prop:spectrum} implies
\[
\|\mathcal B(a,a',b,b')\|=\|H_{s,t}\|=\sqrt{s^2+t^2}=2.
\]
Therefore, for every density operator $\rho$, H\"older's inequality for Schatten norms gives
\[
\bigl|\Tr\bigl(\rho\,\mathcal B(a,a',b,b')\bigr)\bigr|
\le \|\rho\|_1\,\|\mathcal B(a,a',b,b')\| = 2.
\]

This proves the theorem.
\end{proof}

\begin{corollary}
No state on $\mathbb C^3\otimes\mathbb C^3$ violates the CHSH inequality when each party is restricted to spin-$1$ observables of the form $S(u)$ with $\|u\|=1$.
\end{corollary}

\begin{proof}
This is exactly the statement of Theorem~\ref{thm:main}.
\end{proof}

The bound is tight. For example, choosing $a=a'=b=b'=(0,0,1)$ gives
\[
\mathcal B(a,a',b,b') = 2S_z\otimes S_z,
\]
and the product state $\ket{1,1}$ attains expectation value $2$.

\section{Discussion}

We have shown that the CHSH inequality is never violated by bipartite qutrit states when both parties are restricted to spin-\(1\) observables of the form \(S(u)\). The result also identifies the mechanism behind the nonviolation, since every spin-\(1\) CHSH Bell operator is locally unitarily equivalent to a two-parameter operator, and the CHSH geometry forces its operator norm to be exactly the classical bound \(2\). Thus the obstruction is not a property of a particular state, but a structural property of the allowed measurement family. This should not be interpreted as saying that two-qutrit entanglement cannot exhibit Bell nonlocality under other measurements or other Bell inequalities \citep{kaszlikowski2000violations, durt2001violations, collins2002bell, methot2007anomaly, alsina2016operational}. Rather, it shows that local spin-component measurements are too restrictive to reveal CHSH nonlocality, even though they are physically natural observables.

\section*{Acknowledgments}

This work is in part supported by the National Research Foundation of Korea (NRF, RS-2024-00451435 (20\%), RS-2024-00413957 (20\%)), Institute of Information \& communications Technology Planning \& Evaluation (IITP, RS-2025-02305453 (15\%), RS-2025-02273157 (15\%), RS-2025-25442149 (15\%), RS-2021-II211343 (15\%)) grant funded by the Ministry of Science and ICT (MSIT), Institute of New Media and Communications (INMAC), and the BK21 FOUR program of the Education, Artificial Intelligence Graduate School Program (Seoul National University), and Research Program for Future ICT Pioneers, Seoul National University in 2026.

\section*{Author contributions}

H.C. conceptualized the framework and wrote the manuscript. J.L. provided supervision and reviewed the manuscript.

\section*{Data availability}
No data were created or analyzed in this study.

\bibliography{main}

\end{document}